\documentclass[journal]{IEEEtran}

\usepackage{amsfonts}
\usepackage{amssymb}
\usepackage{amsthm}
\usepackage{amsmath,amsfonts,amssymb}
\usepackage[dvips]{graphicx}
\usepackage{subfigure}
\usepackage{verbatim}
\usepackage{setspace}
\usepackage{bm}
\usepackage{algorithmic} 
\renewenvironment{proof}{\noindent\textit{Proof.}}{\hfill $\blacksquare$}
\usepackage[ruled,vlined]{algorithm2e}
\usepackage{cite}

\usepackage{changepage}
\usepackage{pdfpages}
\usepackage{color}
\newtheorem{theorem}{Theorem}
\newtheorem{lemma}{Lemma}

\newcommand{\figwidth}{8}
\IEEEoverridecommandlockouts

\begin{document}
\title{\huge Movable-Antenna Array Enhanced Beamforming:\\ Achieving Full Array Gain with Null Steering}
\author{Lipeng Zhu, ~\IEEEmembership{Member,~IEEE,}
		Wenyan Ma,~\IEEEmembership{Student Member,~IEEE,}
		and Rui Zhang,~\IEEEmembership{Fellow,~IEEE}
	\thanks{This work is supported in part by MOE Singapore under Award T2EP50120-0024, National University of Singapore under Research Grant R-261-518-005-720, and The Guangdong Provincial Key Laboratory of Big Data Computing. (\textit{Corresponding author: Lipeng Zhu})}
	\thanks{L. Zhu and W. Ma are with the Department of Electrical and Computer Engineering, National University of Singapore, Singapore 117583 (e-mail: zhulp@nus.edu.sg, wenyan@u.nus.edu).}
	\thanks{R. Zhang is with School of Science and Engineering, Shenzhen Research Institute of Big Data, The Chinese University of Hong Kong, Shenzhen, Guangdong 518172, China (e-mail: rzhang@cuhk.edu.cn). He is also with the Department of Electrical and Computer Engineering, National University of Singapore, Singapore 117583 (e-mail: elezhang@nus.edu.sg).}
	\vspace{-0.6 cm}
}

\maketitle


\begin{abstract}
	Conventional beamforming with fixed-position antenna (FPA) arrays has a fundamental trade-off between maximizing the signal power (array gain) over a desired direction and simultaneously minimizing the interference power over undesired directions. To overcome this limitation, this letter investigates the movable antenna (MA) array enhanced beamforming by exploiting the new degree of freedom (DoF) via antenna position optimization, in addition to the design of antenna weights. We show that by jointly optimizing the antenna positions vector (APV) and antenna weights vector (AWV) of a linear MA array, the full array gain can be achieved over the desired direction while null steering can be realized over all undesired directions, under certain numbers of MAs and null-steering directions. The optimal solutions for AWV and APV are derived in closed form, which reveal that the optimal AWV for MA arrays requires only the signal phase adjustment with a fixed amplitude. Numerical results validate our analytical solutions for MA array beamforming and show their superior performance to the conventional beamforming techniques with FPA arrays.
\end{abstract}
\begin{IEEEkeywords}
	Movable antenna (MA) array, beamforming, array gain, null steering.
\end{IEEEkeywords}


%
\IEEEpeerreviewmaketitle

\section{Introduction}
\IEEEPARstart{B}{eamforming} is an important signal processing technique for realizing directional signal transmission/reception in multiple-antenna systems. By controlling the amplitude and/or phase of the signal at each antenna, the signal wavefronts to/from different antennas can be constructively superimposed for amplifying signals over desired directions or destructively canceled for eliminating interference over undesired directions \cite{xiao2023array,li2005robust}. Over the past few decades, beamforming techniques have been widely applied in wireless communication, radar, sonar, imaging systems, etc., for fulfilling different performance requirements \cite{xiao2023array,li2005robust,zhu2019millim,zhu2017Kronec}. However, due to the fixed geometry of conventional antenna arrays, i.e., fixed-position antenna (FPA) arrays, the existing beamforming solutions in general face a fundamental trade-off between amplifying signals over desired directions and mitigating interference over undesired directions \cite{xiao2023array,li2005robust,zhu2019millim,zhu2017Kronec}. This is because the steering vectors (SVs) of an FPA array have inherent spatial correlation over different steering angles. As such, the maximum signal power or full array gain over the desired direction and null steering over other undesired directions generally cannot be concurrently achieved with classical beamforming designs such as zero-forcing (ZF) \cite{xiao2023array}. 

Recently, movable antenna (MA) was proposed to enable the local movement of antennas for pursuing more favorable channel conditions and achieving better communication performance \cite{zhu2022MAmodel,zhu2023MAmag}. Preliminary studies have validated that by optimizing the MAs' positions, the spatial diversity and multiplexing performance of MA-aided communication systems can be significantly improved compared to conventional FPA systems \cite{zhu2022MAmodel,zhu2023MAmag,ma2022MAmimo,zhu2023MAmultiuser,wu2023movable}. Moreover, an MA array can also achieve enhanced beamforming over FPA arrays by jointly designing the antenna positions vector (APV) and antenna weights vector (AWV). Although such optimization problems for MA arrays have been previously investigated \cite{ismail1991null,hejres2004null,bevelacqua2007optimizing}, only numerical solutions are provided therein which lack analytical insights. Besides, it was shown in \cite{Leshem2021align} that the interference from multiple directions can be approximately nulled by adjusting the distance of two antennas at the receiver, under the assumption of irrational-valued interference angles and an infinitely large region for placing antennas. In summary, existing literatures have not addressed the fundamental question that if it is possible to achieve the full array gain of an MA array with complete interference nulling by exploiting the new degree of freedom (DoF) in antenna position optimization.

To answer this question, we investigate in this letter the enhanced beamforming of a linear MA array by jointly optimizing its APV and AWV. We analytically show that the full array gain can be reaped over the desired signal direction while null steering can be realized over undesired interference directions with MA arrays, under certain numbers of MAs and null-steering directions. The key idea of our proposed solution lies in that the optimal MAs' positions can transform the geometry of the MA array such that the SV over the desired direction becomes orthogonal to those over all undesired directions. Moreover, the optimal solutions for the corresponding AWV and APV of the MA array are derived in closed form, which reveal that only analog beamforming is needed for MA arrays, i.e., the optimal AWV only requires signal phase adjustment with a fixed amplitude \cite{xiao2023array}, thus significantly reducing the beamforming implementation complexity. Numerical results validate our analytical solutions for MA array beamforming and show their performance superiority to the conventional beamforming techniques with FPA arrays.

\textit{Notation}: $a$, $\mathbf{a}$, and $\mathbf{A}$ denote a scalar, a vector, and a matrix respectively. $(\cdot)^{\mathrm{T}}$, $(\cdot)^{\mathrm{H}}$, and $(\cdot)^{-1}$ denote transpose, conjugate transpose, and inverse, respectively. $\mathbb{Z}$, $\mathbb{R}$, and $\mathbb{C}$ represent the sets of integers, real numbers, and complex numbers, respectively. $|a|$ and $\|\mathbf{a}\|_{2}$ denote the amplitude of scalar $a$ and the 2-norm of vector $\mathbf{a}$, respectively. $[\mathbf{a}]_{n}$ denotes the $n$-th entry of vector $\mathbf{a}$. $\otimes$ represents the Kronecker product.

\section{Problem Formulation and Transformation}
As shown in Fig. \ref{Fig:MA_array}, we consider a linear MA array of size $N$, where the position of the $n$-th antenna is denoted by $x_{n}$, $1 \leq n \leq N$. Denote $\mathbf{x}=\left[x_{1},x_{2},\cdots,x_{N}\right]^{\mathrm{T}} \in \mathbb{R}^{N}$ as the APV of the MA array. The SV of the MA array is thus determined by the APV, which is given by
\begin{equation}\label{eq_steering_vec}
	\mathbf{a}(\mathbf{x}, \theta)=\left[e^{j\frac{2\pi}{\lambda}x_{1}\cos\theta}, e^{j\frac{2\pi}{\lambda}x_{2}\cos\theta}, \cdots, e^{j\frac{2\pi}{\lambda}x_{N}\cos\theta}\right]^{\mathrm{T}},
\end{equation}
where $\lambda$ denotes the wavelength and $\theta$ (in deg) is the steering angle with respect to (w.r.t.) the linear MA array shown in Fig. \ref{Fig:MA_array}.
Denoting $\mathbf{w} \in \mathbb{C}^{N}$ as the AWV for beamforming, the beam pattern of the MA array can be expressed as
\begin{equation}\label{eq_beam_pattern}
	G_{\mathbf{x}, \mathbf{w}}(\theta) = \left|\mathbf{a}^{\mathrm{H}}(\mathbf{x}, \theta)\mathbf{w} \right|^{2},~\theta \in [0^{\circ},180^{\circ}).
\end{equation}

\begin{figure}[t]
	\begin{center}
		\includegraphics[width=\figwidth cm]{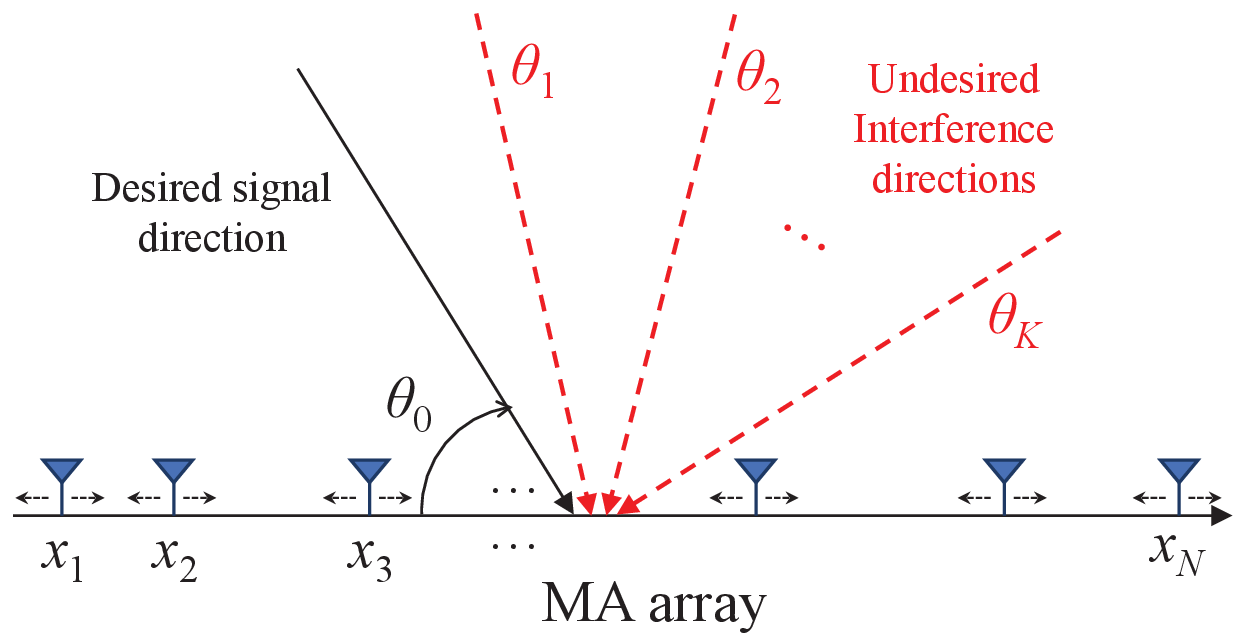}
		\caption{Illustration of the linear MA array and the steering angles.}
		\label{Fig:MA_array}
	\end{center}
\end{figure}

In this letter, we consider the interference dominant scenario (thus, ignoring the receiver noise) where the beam gain should be nulled to zero over all undesired interference directions $\left\{\theta_{k}\right\}_{1 \leq k \leq K}$, where $K$ denotes the total number of null directions. Under this setup, the APV and AWV are jointly optimized for maximizing the beam gain over the desired direction $\theta_{0}$, which can be expressed as the following optimization problem\footnote{An implicit assumption throughout this letter is $\theta_{0} \neq \theta_{k}$ for $\forall k \geq 1$ because the interference cannot be nulled by beamforming if it is incident from exactly the same direction as that of the desired signal. For the case of  $\theta_{0} = \theta_{k}$, interference mitigation can only be implemented by other techniques in the time/frequency domain.}: 
\begin{subequations}\label{eq_problem}
\begin{align}
	\mathop{\mathrm{Max}} \limits _{\mathbf{x}, \mathbf{w}}~~ &G_{\mathbf{x}, \mathbf{w}}(\theta_{0}) \label{eq_problem_a}\\
	\mathrm{s.t.}  ~~~&G_{\mathbf{x}, \mathbf{w}}(\theta_{k}) = 0,~ 1 \leq k \leq K,  \label{eq_problem_b}\\
					&|x_{m}-x_{n}| \geq d_{\min},~  1 \leq m \neq n \leq N,  \label{eq_problem_c} \\
					&\left\|\mathbf{w}\right\|_{2}^{2} = 1,  \label{eq_problem_d}
\end{align}
\end{subequations}
where \eqref{eq_problem_b} is the null-steering constraint; $d_{\min}$ in constraint \eqref{eq_problem_c} is the minimum distance between any two MAs to avoid the coupling effect; and constraint \eqref{eq_problem_d} ensures the normalized power of the AWV. 

It is known that a tight upper bound on the objective function in \eqref{eq_problem_a} is $\bar{G}_{\mathbf{x}, \mathbf{w}}(\theta_{0}) \leq \left\|\mathbf{a}(\mathbf{x}, \theta_{0})\right\|_{2}^{2} \times \left\| \mathbf{w} \right\|_{2}^{2} = N$, which indicates the full array gain over the desired direction. However, due to constraint \eqref{eq_problem_b}, the upper bound cannot be achieved in general with FPA arrays. Specifically, for any given APV $\mathbf{x}$, an optimal solution of the AWV (assumed to be digital beamforming with continuous signal amplitude and phase values) for maximizing $G_{\mathbf{x}, \mathbf{w}}(\theta_{0})$ under the null-steering constraint is given by the ZF beamformer \cite{xiao2023array}, i.e.,
\begin{equation}\label{eq_opt_AWV}
	\begin{aligned}
		&\mathbf{w}^{\mathrm{ZF}}_{\mathbf{x}} = \frac{\mathbf{w}_{\mathbf{x}}}{\left\| \mathbf{w}_{\mathbf{x}} \right\|_{2}},\\
		&\mathbf{w}_{\mathbf{x}}=\left[\mathbf{I}_{N}-\mathbf{A}(\mathbf{x})\left(\mathbf{A}(\mathbf{x})^{\mathrm{H}}\mathbf{A}(\mathbf{x})\right)^{-1} \mathbf{A}(\mathbf{x})^{\mathrm{H}}\right] \mathbf{a}(\mathbf{x}, \theta_{0}),
	\end{aligned}
\end{equation}
with $\mathbf{A}(\mathbf{x})=\left[\mathbf{a}(\mathbf{x}, \theta_{1}),\mathbf{a}(\mathbf{x}, \theta_{2}),\cdots,\mathbf{a}(\mathbf{x}, \theta_{K})\right]$. The resulting beam gain over the desired direction can be obtained as 
\begin{equation}\label{eq_beam_gain}\small
	\begin{aligned}
		&G_{\mathbf{x},\mathbf{w}^{\mathrm{ZF}}_{\mathbf{x}}}(\theta_{0}) = \left|\mathbf{a}^{\mathrm{H}}(\mathbf{x}, \theta_{0}) \mathbf{w}^{\mathrm{ZF}}_{\mathbf{x}} \right|^{2}\\
		&=\left|\mathbf{a}(\mathbf{x}, \theta_{0})^{\mathrm{H}} \left[\mathbf{I}_{N}-\mathbf{A}(\mathbf{x})\left(\mathbf{A}(\mathbf{x})^{\mathrm{H}}\mathbf{A}(\mathbf{x})\right)^{-1} \mathbf{A}(\mathbf{x})^{\mathrm{H}}\right] \mathbf{a}(\mathbf{x}, \theta_{0})\right|\\
		&=N-\underbrace{\mathbf{a}(\mathbf{x}, \theta_{0})^{\mathrm{H}} \mathbf{A}(\mathbf{x})\left(\mathbf{A}(\mathbf{x})^{\mathrm{H}}\mathbf{A}(\mathbf{x})\right)^{-1} \mathbf{A}(\mathbf{x})^{\mathrm{H}} \mathbf{a}(\mathbf{x}, \theta_{0})}_{\triangleq L(\mathbf{x})},
	\end{aligned}
\end{equation}
where $L(\mathbf{x}) \geq 0$ represents the loss of the array gain over the desired direction caused by ZF beamforming for null steering over all undesired directions. For conventional FPA arrays, $\mathbf{A}(\mathbf{x})$, $\mathbf{a}(\mathbf{x}, \theta_{0})$, and $L(\mathbf{x})$ are fixed with given $\mathbf{x}$. Particularly, $L(\mathbf{x})$ increases as the correlation between the SVs over the desired direction and undesired directions becomes higher (see Fig. \ref{Fig:gain_loss} in Section IV for an example). In contrast, for MA arrays, the antenna position optimization for $\mathbf{x}$ offers additional DoFs for decreasing $L(\mathbf{x})$. As such, problem \eqref{eq_problem} can be equivalently transformed into
\begin{subequations}\label{eq_problem2}
	\begin{align}
		\mathop{\mathrm{Min}} \limits _{\mathbf{x}}~~ &L(\mathbf{x}) \label{eq_problem2_a}\\
		\mathrm{s.t.}  ~~~&|x_{m}-x_{n}| \geq d_{\min},~  1 \leq m \neq n \leq N.  \label{eq_problem2_b}
	\end{align}
\end{subequations}

Due to the non-convex forms of $L(\mathbf{x})$ and constraint \eqref{eq_problem2_b} w.r.t. $\mathbf{x}$, problem \eqref{eq_problem2} is a non-convex optimization problem. A straightforward way to solve problem \eqref{eq_problem2} is by exhaustively searching $\mathbf{x}$ subject to \eqref{eq_problem2_b}, which, however, results in an exponential  complexity in terms of $N$ if assuming the search region is bounded. Next, we focus on solving problem \eqref{eq_problem2} under certain values of $N$ and $K$, for which the minimum value of $L(\mathbf{x})$ in \eqref{eq_problem2_a} is zero, i.e., the full array gain of the MA array can be achieved over the desired direction subject to null steering. As such, we consider the following feasibility problem:
\begin{equation}\label{eq_problem3}
	\begin{aligned}
		\mathop{\mathrm{Find}} ~~ &\mathbf{x} \\
		\mathrm{s.t.}  ~~~&L(\mathbf{x})=0,~ \eqref{eq_problem2_b}.
	\end{aligned}
\end{equation}

Since $\left(\mathbf{A}(\mathbf{x})^{\mathrm{H}}\mathbf{A}(\mathbf{x})\right)^{-1}$ is a positive definite matrix, $L(\mathbf{x})=0$ is equivalent to $\mathbf{A}(\mathbf{x})^{\mathrm{H}} \mathbf{a}(\mathbf{x}, \theta_{0})=\mathbf{0}_{K}$, where $\mathbf{0}_{K}$ is a $K$-dimensional vector with all elements being zero. This indicates that the SV over the desired direction $\theta_{0}$ should be orthogonal to those over all undesired directions $\left\{\theta_{k}\right\}_{1 \leq k \leq K}$, i.e.,
\begin{equation}\label{eq_SV_ort}
	\begin{aligned}
		\mathbf{a}(\mathbf{x}, \theta_{k})^{\mathrm{H}} \mathbf{a}(\mathbf{x}, \theta_{0})=0,~1 \leq k \leq K.
	\end{aligned}
\end{equation}
We call \eqref{eq_SV_ort} as the SV orthogonality (SVO) condition for achieving the full array gain over the desired direction subject to null steering. Under this condition, the optimal AWV in \eqref{eq_opt_AWV} can be further simplified as 
\begin{equation}\label{eq_opt_AWV2}
	\begin{aligned}
		&\mathbf{w}^{\star}_{\mathbf{x}} = \frac{\mathbf{a}(\mathbf{x}, \theta_{0})}{\left\| \mathbf{a}(\mathbf{x}, \theta_{0}) \right\|_{2}},
	\end{aligned}
\end{equation}
which has constant-modulus elements and thus can be applied to analog beamforming systems for reducing the implementation complexity of MA arrays.

\section{Antenna Position Optimization}
In this section, we demonstrate the feasibility of problem \eqref{eq_problem3} by finding APV $\mathbf{x}$ to satisfy the SVO condition \eqref{eq_SV_ort} and constraint \eqref{eq_problem2_b} under certain values of $N$ and $K$. To this end, we start from the simple case of $K=1$ in the following lemma.
\begin{lemma} \label{Lemma_1Interf}
	For $K=1$, an APV satisfying the SVO condition \eqref{eq_SV_ort} and constraint \eqref{eq_problem2_b} always exists for any $N \geq 2$.
\end{lemma}
\begin{proof}
	For $K=1$, the SVO condition \eqref{eq_SV_ort} is simplified as
	\begin{equation}\label{eq_SV_cond1}
		\begin{aligned}
			\mathbf{a}(\mathbf{x}, \theta_{1})^{\mathrm{H}} \mathbf{a}(\mathbf{x}, \theta_{0}) 
			=  \sum \limits _{n=1}^{N} e^{j\frac{2\pi}{\lambda}x_{n}\left(\cos\theta_{0}-\cos\theta_{1}\right)} =0.
		\end{aligned}
	\end{equation}
	To satisfy constraint \eqref{eq_problem2_b}, an optimal solution for \eqref{eq_SV_cond1} with $N \geq 2$ is given by
	\begin{equation}\label{eq_solution2}
		\begin{aligned}
			x_{n}^{\star}=(n-1)d,~1 \leq n \leq N, 
		\end{aligned}
	\end{equation}
	where $d=\frac{\left(q_{1}+1/N\right)\lambda}{\left|\cos\theta_{0}-\cos\theta_{1}\right|}$ and $q_{1}$ is the minimum integer that ensures $d \geq d_{\min}$. This thus completes the proof.
\end{proof}

Lemma \ref{Lemma_1Interf} indicates that the full array gain over the desired direction $\theta_{0}$ and null steering over the undesired direction $\theta_{1}$ can be concurrently achieved by optimizing the APV for $N \geq 2$\footnote{In practice, the SVO condition \eqref{eq_SV_ort} cannot be satisfied if $\theta_{0}$ and $\theta_{1}$ are extremely close to each other and the region for antenna deployment has a limited size. In this case, a possible solution is to set the distance of MAs as large as possible for minimizing the SVs' correlation over $\theta_{0}$ and $\theta_{1}$. Then, ZF beamforming in \eqref{eq_opt_AWV} can be used for interference nulling.}. Next, we provide the following lemma to extend the result to the case of multiple undesired directions.

\begin{lemma}\label{Lemma_KInterf}
	If an APV satisfying the SVO condition \eqref{eq_SV_ort} and constraint \eqref{eq_problem2_b} exists for $N=N_{1}$ and $K \leq K_{1}$, then an APV satisfying the SVO condition \eqref{eq_SV_ort} and constraint \eqref{eq_problem2_b} also exists for $N=N_{1}N_{2}$ and $K \leq K_{1}+1$, with any $ N_{2} \geq 2$.
\end{lemma}
\begin{proof}
	See Appendix \ref{App_Lemma_KInterf}.
\end{proof}

\begin{figure}[t]
	\begin{center}
		\includegraphics[width=\figwidth cm]{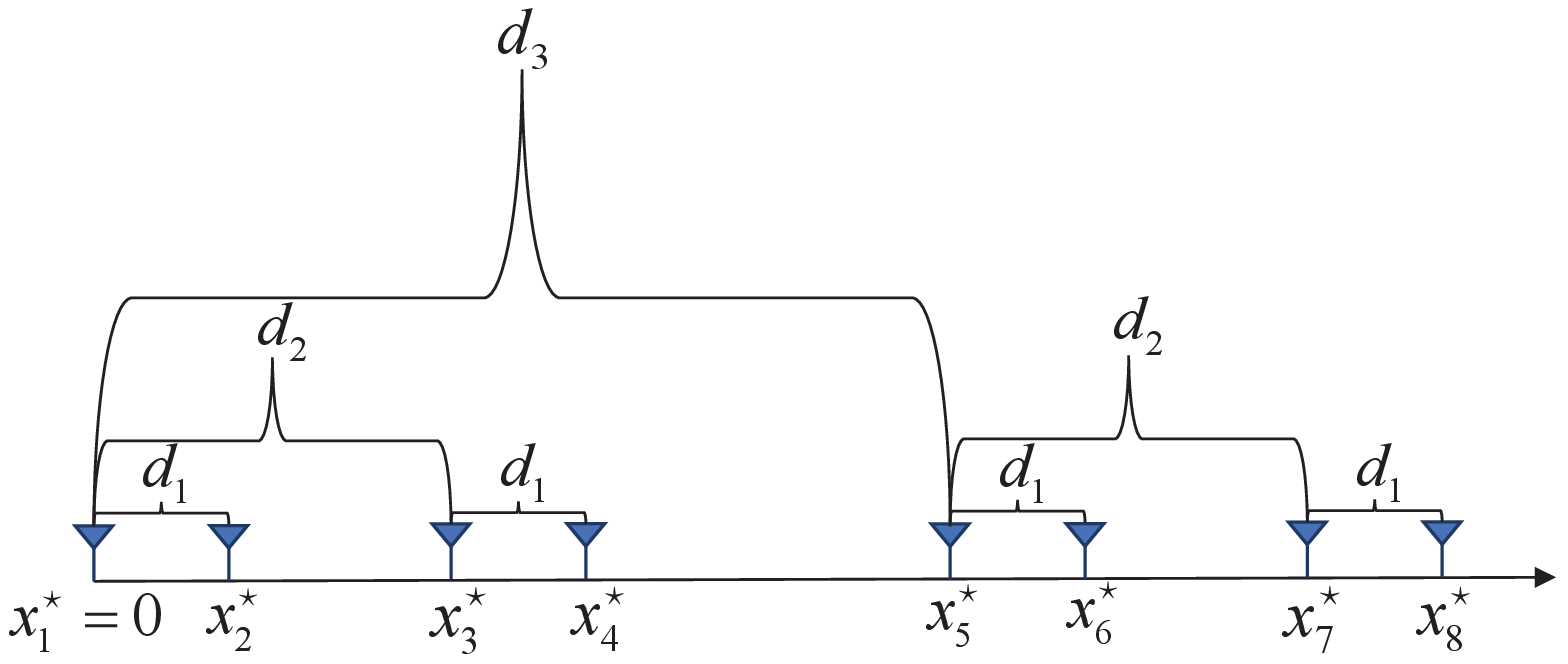}
		\caption{An example of the optimal APV for $N=8$.}
		\label{Fig:opt_position}
	\end{center}
\end{figure}

Next, we are ready to present the feasible solution for problem \eqref{eq_problem3} under some specific values of $N$ and $K$. To this end, we define a factorization vector as follows. Denote the prime factorization of $N$ as $N=\prod _{i=1}^{I_{N}} f_{i}$, where $I_{N}$ represents the total number of prime factors of $N$ and they are sorted in a non-decreasing order $f_{1} \leq f_{2} \leq \cdots \leq f_{I_{N}}$. Then, we define $g_{1}=1$, $g_{i}= \prod_{j=1}^{i-1} f_{j}$, $2 \leq i \leq I_{N}$, and $\mathbf{g}=[g_{1}, g_{2}, \cdots, g_{I_{N}}]^{\mathrm{T}}$. According to basic number theory, for $\forall n \in \mathbb{Z}$, $1 \leq n \leq N$, it can be uniquely determined by the factorization vector $\mathbf{z}_{n} \in \mathbb{Z}^{I_{N}}$ as $n=\mathbf{z}_{n}^{\mathrm{T}} \mathbf{g}+1$ subject to $[\mathbf{z}_{n}]_{i} < f_{i}$, $1 \leq i \leq I_{N}$. Specifically, $\mathbf{z}_{n}$ can be determined by the integer quotients of successively dividing the remainder of number $(n-1)$ by each element in $\mathbf{g}$ (from back to front). For example, for $N=30=2 \times 3 \times 5$, we have $I_{N}=3$ and $\mathbf{g}=[1, 2, 6]^{\mathrm{T}}$. Then, we can obtain $\mathbf{z}_{5}=[0, 2, 0]$, $\mathbf{z}_{24}=[1, 2, 3]$, and so on.

\begin{theorem}\label{Theorem_NMA}
	Denoting the prime factorization of $N$ as $N=\prod _{i=1}^{I_{N}} f_{i}$, an APV $\mathbf{x}^{\star}$ satisfying the SVO condition \eqref{eq_SV_ort} and constraint \eqref{eq_problem2_b} always exists for all $K \leq I_{N}$, which is given by
	\begin{equation}\label{eq_opt_position}
		x_{n}^{\star}= \mathbf{z}_{n}^{\mathrm{T}}\mathbf{d},~1 \leq n \leq N,
	\end{equation}
	with $\mathbf{d}=[d_{1}, d_{2},\cdots, d_{I_{N}}]^{\mathrm{T}}$ and
	\begin{equation}\label{eq_opt_distance}
		d_{i}=\left\{
		\begin{aligned}
			&\frac{\left(q_{i}+1/f_{i}\right)\lambda}{|\cos\theta_{0}-\cos\theta_{i}|},~1\leq i \leq K,\\
			&\sum _{j= 1}^{i-1} (f_{j}-1)d_{j} + d_{\min},~K+1\leq i \leq I_{N},
		\end{aligned}
		\right.
	\end{equation}
	where $q_{i}$ is the minimum integer ensuring $d_{i} \geq \sum _{j= 1}^{i-1} (f_{j}-1)d_{j} + d_{\min}$.
\end{theorem}
\begin{proof}
	See Appendix \ref{App_Theorem_NMA}.
\end{proof}

Theorem \ref{Theorem_NMA} indicates that the full array gain over the desired direction $\theta_{0}$ and null steering over $K$ undesired directions $\{\theta_{k}\}_{1 \leq k \leq K}$ can be concurrently achieved by optimizing the APV when $K \leq I_{N}$. The basic idea for constructing the optimal solution in the proof of Theorem \ref{Theorem_NMA} is by ensuring the SVO condition over undesired directions one by one, subject to the minimum-distance constraint between any two MAs. Fig. \ref{Fig:opt_position} shows an example of the optimal APV for $N=8$ and $K=3$, where the 1st MA is deployed at $x_{1}^{\star}=0$. Then, the 2nd MA is deployed at $x_{2}^{\star}=x_{1}^{\star} + d_{1}$ for satisfying the SVO condition over $\theta_{1}$. Next, the 3rd and 4th MAs are deployed at $x_{3}^{\star}=x_{1}^{\star} + d_{2}$ and $x_{4}^{\star}=x_{2}^{\star} + d_{2}$, respectively, for satisfying the SVO condition over $\theta_{2}$, while the SVO condition over $\theta_{1}$ is still guaranteed. Finally, the 5th-8th MAs are deployed at $x_{5}^{\star}=x_{1}^{\star} + d_{3}$, $x_{6}^{\star}=x_{2}^{\star} + d_{3}$, $x_{7}^{\star}=x_{3}^{\star} + d_{3}$, and $x_{8}^{\star}=x_{4}^{\star} + d_{3}$ for guaranteeing the SVO condition over $\theta_{3}$, while the SVO conditions over $\theta_{1}$ and $\theta_{2}$ are still maintained. Since only integer $q_{i}$ needs to be searched for obtaining $d_{i}$, $1 \leq i \leq I_{N}$, the computational complexity for constructing the optimal APV in Theorem \ref{Theorem_NMA} is $\mathcal{O}(I_{N})$.

It is worth noting that Theorem \ref{Theorem_NMA} in general only provides the sufficient conditions for the optimal APV satisfying the SVO condition \eqref{eq_SV_ort} and constraint \eqref{eq_problem2_b}, which may not be necessary. For the case of $K> I_{N}$, it still remains an open problem whether such an optimal APV solution exists or not to ensure $L(\mathbf{x})=0$. In this case, suboptimal solutions for problem \eqref{eq_problem} can be obtained as follows. First, we select $I_{N}$ prior directions which have strong interference and design the APV according to Theorem \ref{Theorem_NMA}. Then, for the designed APV, we employ ZF beamforming in \eqref{eq_opt_AWV} to null the interference over the other $(K-I_{N})$ undesired directions.

\section{Numerical Results}
In this section, numerical results are provided to validate the enhanced beamforming performance of MA arrays, where we employ the optimal APV given in \eqref{eq_opt_position} and the corresponding optimal AWV shown in \eqref{eq_opt_AWV2}. An $N$-dimensional FPA array with half-wavelength antenna spacing is considered as a benchmark for performance comparison. Specifically, digital beamforming is used for the FPA array, which employs the ZF solution given in \eqref{eq_opt_AWV}. In addition, analog beamforming is also considered for the FPA array by adopting the Kronecker decomposition-based approach proposed in \cite{zhu2017Kronec}, which uses $K$ factors of the analog beamformer with $N$ antennas for nulling interference over all undesired directions and the remaining $(I_{N}-K)$ factors for maximizing the beam gain over the desired direction.

\begin{figure}[t]
	\begin{center}
		\includegraphics[width=\figwidth cm]{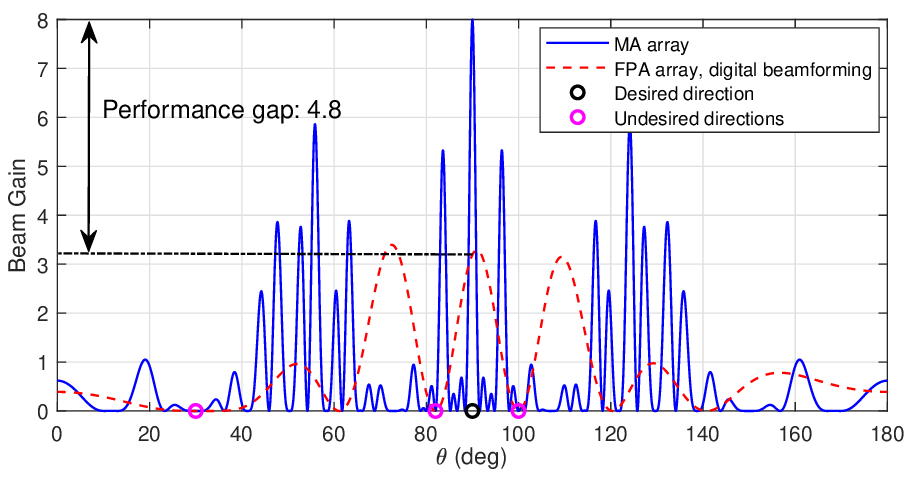}
		\caption{Comparison of the beam patterns between MA and FPA arrays assuming digital beamforming for the FPA array, with $N=8$, $K=3$, $\theta_{0}=90^{\circ}$, $\theta_{1}=30^{\circ}$, $\theta_{2}=82^{\circ}$, and $\theta_{3}=100^{\circ}$.}
		\label{Fig:pattern_ZF}
	\end{center}
\end{figure}

\begin{table}[t]\centering  
	\caption{The normalized APVs of the MA and FPA arrays.}\label{Tab_APV}
	\resizebox{\linewidth}{!}{
	\begin{tabular}{|c|c|c|c|c|c|c|c|c|}
		\hline
		& $x_{1}$ & $x_{2}$ & $x_{3}$ & $x_{4}$ & $x_{5}$ & $x_{6}$ & $x_{7}$ & $x_{8}$ \\ \hline
		FPA 							& 0       & 0.5     & 1       & 1.5     & 2       & 2.5     & 3       & 3.5     \\ \hline
		MA (Fig. \ref{Fig:pattern_ZF})  & 0       & 1.73    & 8.64    & 10.37   & 17.96   & 19.70   & 26.60   & 28.33   \\ \hline
		MA (Fig. \ref{Fig:pattern_kron})  & 0       & 1.52    & 2.66    & 4.18    & 6.10    & 7.63    & 8.76    & 10.29   \\ \hline
	\end{tabular}}
\end{table}

Fig. \ref{Fig:pattern_ZF} compares the beam patterns between the MA and FPA arrays assuming digital beamforming for the later, with $N=8$, $K=3$, $\theta_{0}=90^{\circ}$, $\theta_{1}=30^{\circ}$, $\theta_{2}=82^{\circ}$, and $\theta_{3}=100^{\circ}$. Since the beam gains of the FPA array over all the three undesired directions should be nulled to zero by the ZF-based AWV, its array gain over the desired direction suffers from a significant loss, i.e., $L(\mathbf{x})=4.8$ in \eqref{eq_beam_gain}. In contrast, by exploiting the additional DoFs in antenna position optimization, the SV of the MA array for the desired signal direction becomes orthogonal to those for all the undesired directions. As a result, it is observed that the full array gain (i.e., $N=8$) and null steering are achieved concurrently by the MA array beamforming (with analog beamforming only). The corresponding values of the antenna positions (normalized by $\lambda$) in the obtained APV according to Theorem \ref{Theorem_NMA} for the MA array are shown in Table \ref{Tab_APV}, as compared to those of the FPA array. It is observed that the end-to-end length of the linear MA array corresponding to the optimal APV solution is $x_{8}^{\star}-x_{1}^{\star}=28.33\lambda$, which is about $8$ times longer than that of the FPA array with half-wavelength antenna spacing.

\begin{figure}[t]
	\begin{center}
		\includegraphics[width=\figwidth cm]{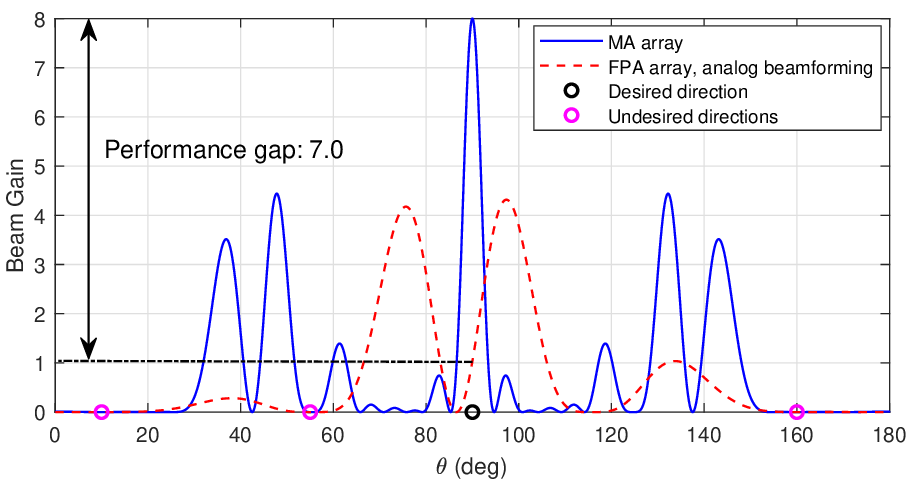}
		\caption{Comparison of the beam patterns between MA and FPA arrays assuming analog beamforming for the FPA array, with $N=8$, $K=3$, $\theta_{0}=90^{\circ}$, $\theta_{1}=10^{\circ}$, $\theta_{2}=55^{\circ}$, and $\theta_{3}=160^{\circ}$.}
		\label{Fig:pattern_kron}
	\end{center}
\end{figure}

Fig. \ref{Fig:pattern_kron} compares the beam patterns between the MA and FPA arrays assuming analog beamforming for the later, with $N=8$, $K=3$, $\theta_{0}=90^{\circ}$, $\theta_{1}=10^{\circ}$, $\theta_{2}=55^{\circ}$, and $\theta_{3}=160^{\circ}$. It is observed again that the full array gain of the MA array and null steering over all three undesired directions are achieved concurrently. In contrast, due to the limited DoF in analog beamforming, the FPA array undergoes significant loss of the array gain ($L(\mathbf{x})=7.0$) over the desired direction when nulling the beam gains over all three undesired directions. Moreover, it is observed from Table \ref{Tab_APV} that the end-to-end length of the linear MA array corresponding to the optimal APV solution is $x_{8}^{\star}-x_{1}^{\star}=10.29\lambda$, which is about $3$ times longer than that of the FPA array with half-wavelength antenna spacing.

\begin{figure}[t]
	\begin{center}
		\includegraphics[width=\figwidth cm]{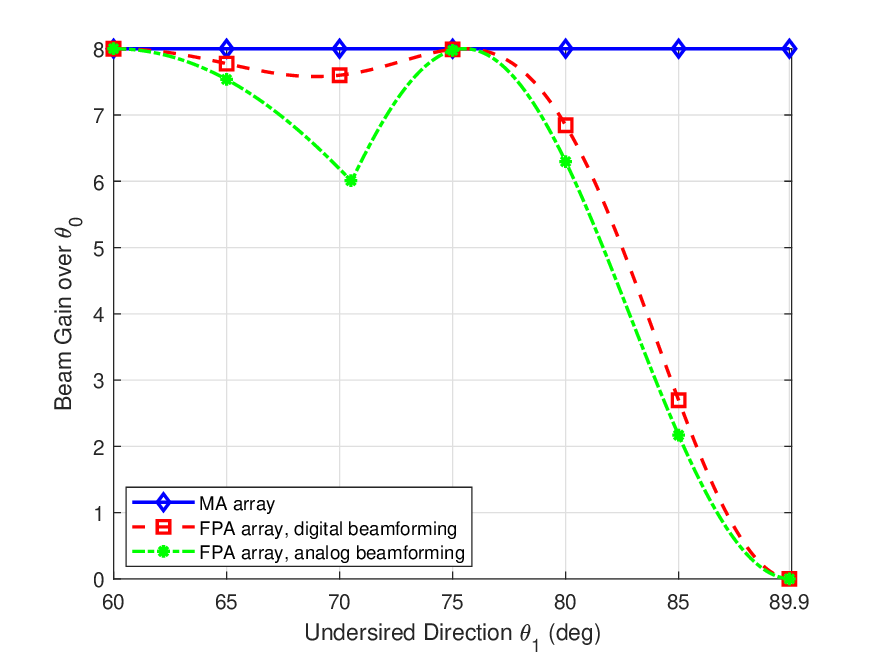}
		\caption{Beam gains of the MA and FPA arrays over the desired direction $\theta_{0}=90^{\circ}$ versus undesired direction $\theta_{1}$, with $N=8$ and $K=1$.}
		\label{Fig:gain_loss}
	\end{center}
\end{figure}

Next, we consider the case of $N=8$ and $K=1$ and evaluate the impact of the undesired direction $\theta_{1}$'s value on the beam gain over the desired signal direction $\theta_{0}=90^{\circ}$. As can be observed from Fig \ref{Fig:gain_loss}, the proposed MA array beamforming can always achieve the full array gain over the desired direction by jointly optimizing the APV and AWV subject to null steering. However, the FPA array cannot achieve the full array gain in general, with digital or analog beamforming. In particular, when $\theta_{1}$ approaches $\theta_{0}$, the loss of the beam gain for the FPA array becomes more significant with both digital and analog beamforming because the correlation of the SVs for the FPA array over the two directions increases.

\section{Conclusion}
In this letter, we investigated the MA array enhanced beamforming by exploiting the new DoF in antenna position optimization. We showed that by jointly designing the APV and AWV of a linear MA array, the full array gain can be reaped over the desired direction with null steering simultaneously realized over all undesired directions, under certain numbers of MAs and null-steering directions. Moreover, the optimal solutions of the APV and AWV for the MA array were derived in closed form, which reveal that MA arrays require analog beamforming with signal phase adjustment only. Numerical results validated our analytical solutions and showed the performance superiority of MA arrays to conventional FPA arrays with digital or analog beamforming.

{\appendices
\section{Proof of Lemma \ref{Lemma_KInterf}} \label{App_Lemma_KInterf}
Let $\left\{\theta_{k}\right\}_{1 \leq k \leq K_{1}+1}$ denote the set of $(K_{1}+1)$ undesired directions. According to the precondition of Lemma \ref{Lemma_KInterf}, there always exists an APV for the $N_{1}$-dimensional MA array, denoted by $\bar{\mathbf{x}}=\left[\bar{x}_{1},\bar{x}_{2},\cdots,\bar{x}_{N_{1}}\right]^{\mathrm{T}}$, which satisfies $\mathbf{a}(\bar{\mathbf{x}}, \theta_{k})^{\mathrm{H}} \mathbf{a}(\bar{\mathbf{x}}, \theta_{0})=0$ for $1 \leq k \leq K_{1}$ as well as $|\bar{x}_{m}-\bar{x}_{n}| \geq d_{\min}$ for $1 \leq m \neq n \leq N_{1}$. Without loss of generality, we assume that the elements in $\bar{\mathbf{x}}$ are sorted in an increasing order, i.e., $\bar{x}_{1}<\bar{x}_{2}<\cdots<\bar{x}_{N_{1}}$. Then, we construct the APV for the $N_{1}N_{2}$-dimensional MA array as
\begin{equation}\label{eq_subarray}
	\begin{aligned}
	\hat{\mathbf{x}}&=\left[\bar{\mathbf{x}}^{\mathrm{T}}, \bar{\mathbf{x}}^{\mathrm{T}}+d, \cdots, \bar{\mathbf{x}}^{\mathrm{T}}+(N_{2}-1)d\right]^{\mathrm{T}},
	\end{aligned}
\end{equation} 
where $d$ is a distance parameter to be determined. Denoting $\mathbf{y}=[0, d, 2d, \cdots, (N_{2}-1)d]^{\mathrm{T}}$, the SVO condition of the constructed APV $\hat{\mathbf{x}}$ can be checked by examining
\begin{equation}\label{eq_solutionN}
	\begin{aligned}
		&\mathbf{a}(\hat{\mathbf{x}}, \theta_{k})^{\mathrm{H}} \mathbf{a}(\hat{\mathbf{x}}, \theta_{0})\\
		=&\left[\mathbf{a}(\mathbf{y}, \theta_{k}) \otimes \mathbf{a}(\bar{\mathbf{x}}, \theta_{k})\right]^{\mathrm{H}} \left[\mathbf{a}(\mathbf{y}, \theta_{0}) \otimes \mathbf{a}(\bar{\mathbf{x}}, \theta_{0})\right] \\
		=&\mathbf{a}(\mathbf{y}, \theta_{k})^{\mathrm{H}} \mathbf{a}(\mathbf{y}, \theta_{0}) \times \mathbf{a}(\bar{\mathbf{x}}, \theta_{k})^{\mathrm{H}} \mathbf{a}(\bar{\mathbf{x}}, \theta_{0})\\
		=&\sum \limits _{n=1}^{N_{2}} e^{j\frac{2\pi}{\lambda}(n-1)d\left(\cos\theta_{0}-\cos\theta_{k}\right)} \times \mathbf{a}(\bar{\mathbf{x}}, \theta_{k})^{\mathrm{H}} \mathbf{a}(\bar{\mathbf{x}}, \theta_{0}).
	\end{aligned}
\end{equation}
Next, we determine the value of $d$ by considering the case of $K = K_{1}+1$ and the case of $K \leq K_{1}$ separately.

\emph{Case 1}: For $K = K_{1}+1$, we consider $d=\frac{\left(q+1/N_{2}\right)\lambda}{\left|\cos\theta_{0}-\cos\theta_{K_{1}+1}\right|}$, with $q$ being the minimum integer which ensures $\bar{x}_{1}+d-\bar{x}_{N_{1}} \geq d_{\min}$. 

On one hand, for $1 \leq k \leq K_{1}$, $\mathbf{a}(\bar{\mathbf{x}}, \theta_{k})^{\mathrm{H}} \mathbf{a}(\bar{\mathbf{x}}, \theta_{0})=0$ ensures $\mathbf{a}(\hat{\mathbf{x}}, \theta_{k})^{\mathrm{H}} \mathbf{a}(\hat{\mathbf{x}}, \theta_{0})=0$. For $k=K_{1}+1$, it is easy to verify $\sum \limits _{n=1}^{N_{2}} e^{j\frac{2\pi}{\lambda}(n-1)d\left(\cos\theta_{0}-\cos\theta_{K_{1}+1}\right)}=0$, and then we have $\mathbf{a}(\hat{\mathbf{x}}, \theta_{K_{1}+1})^{\mathrm{H}} \mathbf{a}(\hat{\mathbf{x}}, \theta_{0})=0$. Thus, we can conclude that $\hat{\mathbf{x}}$ satisfies the SVO condition, i.e., $\mathbf{a}(\hat{\mathbf{x}}, \theta_{k})^{\mathrm{H}} \mathbf{a}(\hat{\mathbf{x}}, \theta_{0})=0$, $1 \leq k \leq K_{1}+1$. 

On the other hand, since $\bar{\mathbf{x}}$ satisfies constraint \eqref{eq_problem2_b}, we have $|[\hat{\mathbf{x}}]_{tN_{1}+m}-[\hat{\mathbf{x}}]_{tN_{1}+n}| = |\bar{x}_{m}-\bar{x}_{n}| \geq d_{\min}$ for $1 \leq m \neq n \leq N_{1}$ and $0 \leq t \leq N_{2}-1$. Besides, for $1 \leq m, n \leq N_{1}$ and $0 \leq t_{1} < t_{2} \leq N_{2}-1$, we have $|[\hat{\mathbf{x}}]_{t_{1}N_{1}+m}-[\hat{\mathbf{x}}]_{t_{2}N_{1}+n}| \geq |[\hat{\mathbf{x}}]_{(t_{1}+1)N_{1}}-[\hat{\mathbf{x}}]_{t_{2}N_{1}+1}| \geq \bar{x}_{1}+d-\bar{x}_{N_{1}} \geq d_{\min}$ due to the increasing order of the elements in $\hat{\mathbf{x}}$. Thus, we conclude that $\hat{\mathbf{x}}$ also satisfies constraint \eqref{eq_problem2_b}, i.e., $|[\hat{\mathbf{x}}]_{m}-[\hat{\mathbf{x}}]_{n}| \geq d_{\min}$, $1 \leq m \neq n \leq N_{1}N_{2}$.

\emph{Case 2}: For $K \leq K_{1}$, we consider $d=\bar{x}_{N_{1}}+d_{\min}-\bar{x}_{1}$. 

On one hand, for $1 \leq k \leq K \leq K_{1}$, we always have $\mathbf{a}(\bar{\mathbf{x}}, \theta_{k})^{\mathrm{H}} \mathbf{a}(\bar{\mathbf{x}}, \theta_{0})=0$ and thus $\mathbf{a}(\hat{\mathbf{x}}, \theta_{k})^{\mathrm{H}} \mathbf{a}(\hat{\mathbf{x}}, \theta_{0})=0$. On the other hand, we always have $|[\hat{\mathbf{x}}]_{m}-[\hat{\mathbf{x}}]_{n}| \geq d_{\min}$, $1 \leq m \neq n \leq N_{1}N_{2}$, which can be proved in a similar way to that for the previous case of $K = K_{1}+1$. Thus, we conclude that $\hat{\mathbf{x}}$ satisfies the SVO condition and constraint \eqref{eq_problem2_b} for $K \leq K_{1}$.

Combining both cases of $K = K_{1}+1$ and $K \leq K_{1}$, we have shown that an APV satisfying the SVO condition \eqref{eq_SV_ort} and constraint \eqref{eq_problem2_b} always exists for $N=N_{1}N_{2}$ and $K \leq K_{1}+1$. This thus completes the proof.

\vspace{-0.1 cm}
\section{Proof of Theorem \ref{Theorem_NMA}} \label{App_Theorem_NMA}
According to Lemma \ref{Lemma_1Interf}, Theorem \ref{Theorem_NMA} holds for $I_{N}=1$ and $K \leq I_{N}$. Then, suppose that the SVO condition \eqref{eq_SV_ort} and constraint \eqref{eq_problem2_b} can be satisfied by an APV for $I_{N_{1}}=K_{1}$ and $K \leq I_{N_{1}}$. According to Lemma \ref{Lemma_KInterf}, it follows that an APV satisfying the SVO condition \eqref{eq_SV_ort} and constraint \eqref{eq_problem2_b} also exists for $I_{N'}=K_{1}+1$ and $K \leq I_{N'}$. This is because we can always rewrite the prime factorization as $N'=\prod _{i=1}^{K_{1}+1} f_{i} = \prod _{i=1}^{K_{1}} f_{i} \times f_{K_{1}+1} \triangleq N_{1} \times N_{2}$ and apply the constructed APV in the proof of Lemma \ref{Lemma_KInterf}. As such, the complete induction ensures that an APV satisfying the SVO condition \eqref{eq_SV_ort} and constraint \eqref{eq_problem2_b} always exists for all $ K \leq I_{N}$. 

Next, we construct the optimal APV based on the above procedure. According to the proof of Lemma \ref{Lemma_KInterf}, the distance between the $(g_{i}+1)$-th antenna and the 1st antenna is given by $d_{i}$. For $1 \leq i \leq K$, $d_{i}=\frac{\left(q_{i}+1/f_{i}\right)\lambda}{|\cos\theta_{0}-\cos\theta_{i}|}$ guarantees the SVO condition over $\theta_{i}$ as well as the minimum distance constraint between the $g_{i}$-th antenna and the $(g_{i}+1)$-th antenna. For $K+1 \leq i \leq I_{N}$, there is no additional null-steering direction and thus $d_{i}=\sum _{1\leq j \leq i-1} (f_{j}-1)d_{j} + d_{\min}$ guarantees the minimum distance constraint between the $g_{i}$-th antenna and the $(g_{i}+1)$-th antenna. Thus, constraint \eqref{eq_problem2_b} is satisfied.

Recall that $\mathbf{z}_{n}$ is a unique $I_{N}$-dimensional vector satisfying $n=\mathbf{z}_{n}^{\mathrm{T}} \mathbf{g}+1$. Denoting $p_{i}=\sum _{j=1}^{i} [\mathbf{z}_{n}]_{j} g_{j}+1$, $1 \leq i \leq I_{N}$, it is easy to verify $p_{I_{N}}=n$ and $ p_{i}-p_{i-1}=[\mathbf{z}_{n}]_{i} g_{i}$, $2 \leq i \leq I_{N}$. Note that according to the constructed APV in the proof of Lemma \ref{Lemma_KInterf}, the distance between the $p_{i}$-th antenna and the $p_{i-1}$-th antenna is $[\mathbf{z}_{n}]_{i} d_{i}$, $2 \leq i \leq I_{N}$. Thus, the distance between the $n$-th antenna and the 1st antenna can be expressed as 
\begin{equation}\label{eq_distance}
	\begin{aligned}
		x_{n}^{\star}-x_{1}^{\star}&= \sum \limits _{i = 2}^{I_{N}} (x_{p_{i}}^{\star}-x_{p_{i-1}}^{\star}) + x_{p_{1}}^{\star}-x_{1}^{\star}\\
		&=\sum _{i = 2}^{I_{N}} [\mathbf{z}_{n}]_{i}d_{i} + [\mathbf{z}_{n}]_{1}d_{1} = \mathbf{z}_{n}^{\mathrm{T}}\mathbf{d} ,~1 \leq n \leq N.
	\end{aligned}
\end{equation}
Without loss of optimality, we set $x_{1}^{\star}=0$. Thus, $\{x_{n}^{\star}\}_{1\leq n \leq N}$ in \eqref{eq_opt_position} satisfies the SVO condition \eqref{eq_SV_ort} and constraint \eqref{eq_problem2_b}, which is an optimal APV of problem \eqref{eq_problem2} for achieving the full array gain over the desired direction. This thus completes the proof.
}

\vspace{-0.3 cm}

\bibliographystyle{IEEEtran} 
\bibliography{IEEEabrv,ref_zhu}

\begin{thebibliography}{10}
\providecommand{\url}[1]{#1}
\csname url@samestyle\endcsname
\providecommand{\newblock}{\relax}
\providecommand{\bibinfo}[2]{#2}
\providecommand{\BIBentrySTDinterwordspacing}{\spaceskip=0pt\relax}
\providecommand{\BIBentryALTinterwordstretchfactor}{4}
\providecommand{\BIBentryALTinterwordspacing}{\spaceskip=\fontdimen2\font plus
\BIBentryALTinterwordstretchfactor\fontdimen3\font minus
  \fontdimen4\font\relax}
\providecommand{\BIBforeignlanguage}[2]{{%
\expandafter\ifx\csname l@#1\endcsname\relax
\typeout{** WARNING: IEEEtran.bst: No hyphenation pattern has been}%
\typeout{** loaded for the language `#1'. Using the pattern for}%
\typeout{** the default language instead.}%
\else
\language=\csname l@#1\endcsname
\fi
#2}}
\providecommand{\BIBdecl}{\relax}
\BIBdecl

\bibitem{xiao2023array}
Z.~Xiao, L.~Zhu, L.~Bai, and X.-G. Xia, \emph{Array Beamforming Enabled
  Wireless Communications}.\hskip 1em plus 0.5em minus 0.4em\relax Boca Raton,
  Florida, USA: CRC Press, 2023.

\bibitem{li2005robust}
J.~Li and P.~Stoica, \emph{Robust adaptive beamforming}.\hskip 1em plus 0.5em
  minus 0.4em\relax Hoboken, New Jersey, USA: John Wiley \& Sons, 2005.

\bibitem{zhu2019millim}
L.~Zhu, J.~Zhang, Z.~Xiao, X.~Cao, D.~O. Wu, and X.-G. Xia, ``Millimeter-wave
  {NOMA} with user grouping, power allocation and hybrid beamforming,''
  \emph{IEEE Trans. Wireless Commun.}, vol.~18, no.~11, pp. 5065--5079, Nov.
  2019.

\bibitem{zhu2017Kronec}
G.~Zhu, K.~Huang, V.~K.~N. Lau, B.~Xia, X.~Li, and S.~Zhang, ``Hybrid
  beamforming via the kronecker decomposition for the millimeter-wave massive
  {MIMO} systems,'' \emph{IEEE J. Select. Areas Commun.}, vol.~35, no.~9, pp.
  2097--2114, Sep. 2017.

\bibitem{zhu2022MAmodel}
L.~Zhu, W.~Ma, and R.~Zhang, ``Modeling and performance analysis for movable
  antenna enabled wireless communications,'' \emph{arXiv preprint: 2210.05325},
  \url{https://arxiv.org/abs/2210.05325}, accessed on 11 Oct. 2022.

\bibitem{zhu2023MAmag}
{L. Zhu, W. Ma, and R. Zhang}, ``Movable antennas for wireless communication:
  Opportunities and challenges,'' \emph{arXiv preprint: 2306.02331},
  \url{https://arxiv.org/abs/2306.02331}, accessed on 4 June 2023.

\bibitem{ma2022MAmimo}
W.~Ma, L.~Zhu, and R.~Zhang, ``{MIMO} capacity characterization for movable
  antenna systems,'' \emph{IEEE Transactions on Wireless Communications}, 2023
  (Early access), DOI: 10.1109/TWC.2023.3307696.

\bibitem{zhu2023MAmultiuser}
L.~Zhu, W.~Ma, B.~Ning, and R.~Zhang, ``Movable-antenna enhanced multiuser
  communication via antenna position optimization,'' \emph{arXiv preprint:
  2302.06978}, \url{https://arxiv.org/abs/2302.06978}, accessed on 14 Feb.
  2023.

\bibitem{wu2023movable}
Y.~Wu, D.~Xu, D.~W.~K. Ng, W.~Gerstacker, and R.~Schober, ``Movable
  antenna-enhanced multiuser communication: Optimal discrete antenna
  positioning and beamforming,'' \emph{arXiv preprint:2308.02304},
  \url{https://arxiv.org/abs/2210.05325}, accessed on 4 Aug. 2023.

\bibitem{ismail1991null}
T.~Ismail and M.~M. Dawoud, ``Null steering in phased arrays by controlling the
  element positions,'' \emph{IEEE Trans. Antennas Propagat.}, vol.~39, no.~11,
  pp. 1561--1566, Nov. 1991.

\bibitem{hejres2004null}
J.~A. Hejres, ``Null steering in phased arrays by controlling the positions of
  selected elements,'' \emph{IEEE Trans. Antennas Propagat.}, vol.~52, no.~11,
  pp. 2891--2895, Nov. 2004.

\bibitem{bevelacqua2007optimizing}
P.~J. Bevelacqua and C.~A. Balanis, ``Optimizing antenna array geometry for
  interference suppression,'' \emph{IEEE Trans. Antennas Propagat.}, vol.~55,
  no.~3, pp. 637--641, Mar. 2007.

\bibitem{Leshem2021align}
A.~Leshem and U.~Erez, ``The interference channel revisited: Aligning
  interference by adjusting antenna separation,'' \emph{IEEE Transactions on
  Signal Processing}, vol.~69, pp. 1874--1884, Mar. 2021.

\end{thebibliography}

\end{document}